\documentclass{article}
\usepackage[T1]{fontenc}
\usepackage[latin9]{inputenc}
\usepackage{color}
\usepackage{amsmath}
\usepackage{amssymb}

\makeatletter
\newenvironment{lyxcode}
{\par\begin{list}{}{
\setlength{\rightmargin}{\leftmargin}
\setlength{\listparindent}{0pt}% needed for AMS classes
\raggedright
\setlength{\itemsep}{0pt}
\setlength{\parsep}{0pt}
\normalfont\ttfamily}%
 \item[]}
{\end{list}}

\usepackage{amsthm}
\newtheorem{theorem}{Theorem}
\newtheorem{lemma}{Lemma}
\newtheorem{corollary}{Corollary}

\makeatother

\begin{document}

\title{Inferring Program Transformations from Type Transformations for Partitioning of Ordered Sets}

\author{Wim Vanderbauwhede}

\maketitle

\begin{abstract}
In this paper I introduce a mechanism to derive program transformations
from order-preserving transformations of vector types. The purpose
of this work is to allow automatic generation of correct-by-construction
instances of programs in a streaming data processing paradigm suitable
for FPGA processing. We show that for it is possible to automatically
derive instances for programs based on combinations of opaque element-processing
functions combined using \emph{foldl} and \emph{map}, purely from
the type transformations. 
\end{abstract}

\section{Introduction}

In this discussion paper I want to introduce a set of \emph{type transformations}
on \emph{vector types}. In this work we will use a simple form of dependent types \cite{bove2009dependent}, but the concept can be generalised to transformations
on other types, including session types \cite{honda2008multiparty}. The overall idea is to introduce
the transformations and then explore the effect of transforming the
types in a program on the program itself, i.e. what are the required
corresponding functions that will transform the types of the computations
while preserving the results of the computations.

The purpose of this work is to allow automatic generation of correct-by-construction instances of programs in a streaming data processing paradigm suitable for data processing using FPGAs (Field Programmable Gate Arrays, \cite{vanderbauwhede2013high}). Using an optimisation technique such as simulated annealing \cite{aarts1988simulated} and a cost model for the FPGA implementation, the best instance can be automatically selected. 

\section{Preliminaries}

\subsection{Type Variables}
\begin{description}
\item [{\emph{a}}] is a type variable representing a nullary type constructor.
We will call this kind of type variable \emph{atomic}.
\item [{\emph{b},\emph{c}}] are general type variables. We call the set
of these type variables $\mathcal{B}$
\item [{\emph{k},\emph{l,m},\emph{n}}] are non-zero natural numbers, i.e.
$k,m,n\in\mathbb{N}_{>0}$. We will call these \emph{sizes}.
\item [{\emph{p},\emph{q}}] are (unary) type constructor variables, i.e.
types that apply to other types and can depend on non-zero natural
numbers, e.g. \emph{p n a} or \emph{p k q q b}. Note that we use unary
(i.e. one type and one size), right-associative type constructors
purely to simplify the discussion, not as a fundamental limitation.
The crucial point is however that these are \emph{dependent types}. 
\item [{\emph{F},\emph{G},\emph{H}}] are general functions operating on
types (we could call them type transformers). We assume the functions
take a single type as argument and are right associative. Consequently,
we can write \emph{G (F a)} as \emph{G F a}
\item [{\emph{S},\emph{M},\emph{R}~and~\emph{I}}] are specific functions
operating on types, to de defined later.
\end{description}

\subsection{Notations and Definitions}
\begin{description}
\item [{\emph{total~size}}] The \emph{total size} of a type is the product
of all sizes:
\end{description}
$\mathcal{N}(p_{1}\, n_{1}\, p_{2}\, n_{2}...p_{i}...p_{k}\, n_{k}\; a)\overset{{\scriptscriptstyle \triangle}}{=}\prod_{i=1}^{k}n_{i}$
\begin{description}
\item [{\emph{type~transformation}}] A \emph{type transformation} is the
application of a function from one type to another to a type, i.e.
$F\, b\,=\, c$
\item [{\emph{i},\emph{j}}] are used as subscripts to distinguish between
type variables of the same class, so that we can write $p_{1}\; p_{2}...p_{i}...p_{k}\; a$
\end{description}

\section{Restrictions on Type Transformations\label{Restrictions-on-Type}}

Given a type expression of the form $p_{1}\, n_{1}\, p_{2}\, n_{2}...p_{i}...p_{k}\, n_{k}\; a$,
the type transformations we want to consider must obey following restrictions:
\begin{enumerate}
\item The transformations do not remove any atomic type variables or introduce
fresh atomic type variables. It follows that atomic type variables
cannot be modified either. 
\item The transformations can only remove or add one or more outer type
constructors. 
\item The type transformations can transform the \emph{sizes}, but the total
size of the type is an invariant. 
\end{enumerate}
Although the types and transformations are more general, our focus
is on transformations of types describing ordered sets. The above
restrictions intend to reflect that the type transformations should
not alter the number or nature of the elements of the set, but only
the way the set is partitioned.

\section{Vector Types}

For general types $p_{i}$ it may be hard to prove that the above
rules do not alter the number or nature of the elements of the set,
but only the way the set is partitioned. However, if we assume a single
type representing a vector, then what these restrictions say is that
a vector can only be reshaped but not modified in terms of its type
or size. This is of course not a sufficient condition to guarantee
that the type transformations will not change the computations, but
it is a necessary one.

We introduce the \emph{vector type} $\upsilon\, k\, b$ where \emph{k}
is a non-zero positive integer and \emph{b} is an arbitrary type.
This the type representing a vector of length \emph{k} containing
values of type \emph{b}. Specifically we define

\[
b\overset{{\scriptscriptstyle \triangle}}{=}\upsilon\,0\, b
\]

Given an atomic type \emph{a}, we can generate the set of all vector
types $V(a)$ for $a$ as follows:

\[
\begin{cases}
a\in V(a)\\
\forall\, b\,\in\, V(a),\forall\, k\, in\,\mathbb{N}_{>0}|\upsilon\, k\, b\in V(a)
\end{cases}
\]

For convenience, we introduce the following notation: 

\emph{$\upsilon\, k\, b\overset{not.}{=}[b]\langle k\rangle$}

And the shorthand:

\emph{$[...[[b]\langle k_{1}\rangle]\langle k_{2}\rangle...]\langle k_{n}\rangle\overset{not.}{=}[b]\langle k_{1}\rangle\langle k_{2}\rangle...\langle k_{n}\rangle$}

\section{Transformations on Vector Types }

For the rest of the paper we consider a specific case of types and
type transformations: transformations on \emph{vector types}. I posit
three fundamental transformations, each with a corresponding inverse:
\begin{itemize}
\item converting a type to a singleton vector type
\item applying a type transformation to the type variable of a vector type
(\emph{mapping})
\item reshaping a vector type, i.e. modifying the sizes in a vector type
such that the total size is remains invariant
\end{itemize}
We can formalise each of these transformations:

\subsection{Singleton vector type }

The purpose of this operation is to change the dimensionality of a
vector. 

$S\, b\,\overset{{\scriptscriptstyle \triangle}}{=}\,[b]\langle1\rangle$

The inverse operation (reducing dimensionality) is defined trivially
as

$S^{-1}\,[b]\langle1\rangle\,\overset{{\scriptscriptstyle \triangle}}{=}\, b$

So that

$S^{-1}S\, b\,=\, b$

$S\, S^{-1}\,[b]\langle1\rangle\,=\,[b]\langle1\rangle$

Repeated application of S leads to higher-dimensional singleton vectors:

$S\, S\, b=[[b]\langle1\rangle]\langle1\rangle$

We introduce a convenient notation

$S^{k}\, b=[b]\langle1\rangle^{k}$

\subsection{Mapping}

Mapping applies a transformation to the type argument of a vector
type.

This operation is independent of the size, so I have omitted it:

$M\, F\,[b]\,\overset{{\scriptscriptstyle \triangle}}{=}\,[F\, b]$

Note that the inverse operation is the application of the inverse
of \emph{F}, not of \emph{M}:

$M\, F^{-1}\,[F\, b]\,=\,[b]$

Although of course we can define purely notationally

$M^{-1}\, F\overset{{\scriptscriptstyle \triangle}}{=}M\, F^{-1}$

Repeated application of \emph{M} has two cases. The first case is
applying different transformations to a 1-D vector

$(M\, F)\,(M\, G)\,[b]=[F\, G\, b]$

We can rewrite the lhs as

$(M\, F)\,(M\, G)\,[b]=M\,(F\, G)\,[b]$

The second case is applying a single transformation to a multi-dimensional
vector

$M\,(M\, F)\,[[b]]=[M\, F\,[b]]=[[F\, b]]$

We can rewrite the lhs as

$M\,(M\, F)\,[[b]]=M^{2}\, F\,[[b]]$

\subsection{Reshaping}

The purpose of this operation is to re-partition a vector. The operation
works on a 2-D vector.

$R\, m\,[[b]\langle n_{1}\rangle]\langle n_{2}\rangle\overset{{\scriptscriptstyle \triangle}}{=}[[b]\langle n_{1}.m\rangle]\langle n_{2}/m\rangle$

The condition on \emph{m} is of course that \emph{n/m} is a natural
number, i.e. $n_{2}$ is a multiple of \emph{m}.

The inverse operation can again be defined notationally:

$R\, m^{-1}\,[[b]\langle n_{1}\rangle]\langle n_{2}\rangle=[[b]\langle n_{1}/m\rangle]\langle n_{2}.m\rangle$

and

$R^{-1}\, m\overset{{\scriptscriptstyle \triangle}}{=}R\, m^{-1}$

\subsection{Identity Operation}

We define $I\, b\overset{{\scriptscriptstyle \triangle}}{=}b$ for
completeness. I contend (but have not formally proven) that the set
of operations \emph{S},\emph{M},\emph{R},\emph{I} form a group over
V(a). Each of the operations is associative and can be inverted, and
any combination of operations on a vector type results in a vector
type, i.e. it is closed as well. By adding \emph{I}, the conditions
for a group are satisfied.

In fact, as we shall show below, \emph{S},\emph{M},\emph{R},\emph{I}
form a group over a particular finite subset of \emph{V(a)}: 
\begin{itemize}
\item define \emph{V(a,n)} as the subset of \emph{V(a)} where, for any given
vector $[a]\langle k_{1}\rangle\langle k_{2}\rangle...\langle k_{m}\rangle$,
$\prod_{i=1}^{m}k_{i}=n$.
\item then \emph{S,M,R,I} form a group over $V(a,n),\forall a,\, n$
\end{itemize}
In the next section, we will give a proof of the closure constraint.

\subsection{Operations on Atomic Types}

We define mapping on or reshaping of an atomic type as identity operations:

\[
M\, F\, a\overset{{\scriptscriptstyle \triangle}}{=}a
\]

\[
R\, k\, a\overset{{\scriptscriptstyle \triangle}}{=}a
\]

\subsection{Vector Creation}

For what follows, we will need an invertible operation $V$ to create
vector types:

$V\, k\, b\overset{{\scriptscriptstyle \triangle}}{=}[b]\langle k\rangle$

with its inverse

$V^{-1}\, k\,[b]\langle k\rangle\overset{{\scriptscriptstyle \triangle}}{=}b$

In other words, V is equivalent to the vector type constructor but
has an inverse. We use this operation to formally extract the argument
of a vector type from the constructor. The dependent variable k is
not strictly speaking necessary:

$V^{-1}\, k\,(\upsilon\, k\, b)=b$

However, for any vector type $c$, $V^{-1}\, k\, c$ will result in
a type error unless $c=\upsilon\, k\, b$.

\subsection{Theorem: \emph{V(a,n)} is closed and complete under \emph{S,M,R}}

\begin{theorem}

Any type transformation on any vector type in \emph{V(a,N)} that observes
the rules from Section \ref{Restrictions-on-Type}:
\begin{enumerate}
\item can be expressed as a combination of the operations $S$, $M$ and
$R$, and 
\item results in a vector type in \emph{V(a,N).}
\end{enumerate}
\end{theorem}

\begin{proof}

~
\begin{itemize}
\item The most general expression for a type in our system is and multidimensional
vector of type \emph{a}, where the size in every dimension is different.
We consider two instances of this type:

\begin{eqnarray*}
T_{1=}[a]\langle n_{1}\rangle\langle n_{2}\rangle...\langle n_{i}\rangle...\langle n_{k}\rangle\\
T_{2}=[a]\langle m_{1}\rangle\langle m_{2}\rangle...\langle m{}_{j}\rangle...\langle m_{l}\rangle
\end{eqnarray*}

where

\[
\prod_{i=1}^{k}n_{i}=\prod_{j=1}^{l}m_{j}=N
\]

and in general $k\neq l$ and the various $n_{i}$ and $m_{j}$ values
can be different.

\item We aim to show that $T_{1}$ can be transformed into $T_{2}$ through
application of a combination of the operations \emph{S}, \emph{M}
and \emph{R}. Our approach is to first reduce $T_{1}$ to a one-dimensional
vector of size N, and then transform this vector into $T_{2}$.
\item First we reduce the expression for $T_{1}$ to $[a]\langle N\rangle$as
follows

\begin{enumerate}
\item First reshape the outer two vectors through application of \emph{R}:

\begin{eqnarray*}
 &  & R^{-1}\, m_{k-1}\; T_{1}\\
= &  & R^{-1}\, m_{k-1}\;[a]\langle n_{1}\rangle\langle n_{2}\rangle...\langle n_{i}\rangle...\langle n_{k-1}\rangle\langle n_{k}\rangle\\
= &  & [a]\langle n_{1}\rangle\langle n_{2}\rangle...\langle n_{i}\rangle...\langle1\rangle\langle n_{k-1}.n_{k}\rangle
\end{eqnarray*}

which can also be written as

\[
[[[a]\langle n_{1}\rangle\langle n_{2}\rangle...\langle n_{i}\rangle...\langle n_{k-2}\rangle]\langle1\rangle]\langle n_{k-1}.n_{k}\rangle
\]

\item Then apply $S^{-1}$ to the type of the outer vector:

\begin{eqnarray*}
 &  & M\, S^{-1}\,[[[a]\langle n_{1}\rangle\langle n_{2}\rangle...\langle n_{i}\rangle...\langle n_{k-2}\rangle]\langle1\rangle]\langle n_{k-1}.n_{k}\rangle\\
= &  & [S^{-1}\,[[a]\langle n_{1}\rangle\langle n_{2}\rangle...\langle n_{i}\rangle...\langle n_{k-2}\rangle]\langle1\rangle]\langle n_{k-1}.n_{k}\rangle\\
= &  & [[a]\langle n_{1}\rangle\langle n_{2}\rangle...\langle n_{i}\rangle...\langle n_{k-2}\rangle]\langle n_{k-1}.n_{k}\rangle
\end{eqnarray*}

\item Repeating these steps results in

\[
[a]\langle n_{1}.n_{2}.\ldots.n_{i}.\,...\,.n_{k}\rangle
\]

which can be written as

\[
[a]\langle N\rangle
\]

\end{enumerate}
\item Then we perform the reverse process to obtain $T_{2}$:

\begin{enumerate}
\item First increase the dimensionality by calling \emph{S} on the type
of the outer vector

\[
M\, S\,[a]\langle N\rangle=[[a]\langle1\rangle]\langle N\rangle
\]

\item Then reshape the outer two vectors using \emph{R}

\[
R\, m_{1}\,[[a]\langle1\rangle]\langle N\rangle=[[a]\langle m_{1}\rangle]\langle N/m_{1}\rangle
\]

\item Repeating these steps results in
\end{enumerate}

\[
[a]\langle m_{1}\rangle\langle m_{2}\rangle...\langle m{}_{j}\rangle...\langle m_{l}\rangle=T_{2}
\]

\end{itemize}
\end{proof}

\begin{corollary} 
An type transformation consisting of a combination of the transformations \emph{S},\emph{M} and \emph{R} is reversible.
\end{corollary} 

\section{Program Transformations}

In this section we want to explore how transforming a top-level type
impacts on the program. The context is the  FPGA architecture developed for the TyTra project\footnote{http://tytra.org.uk/}, which is similar to the MORA architecture \cite{chalamalasetti2009mora}, and
we consider a simple pipeline of computations. Like MORA, and indeed most FPGA architectures, TyTra assumes that
data is streamed, and we model this as a map over a vector.

I will use the notation $t_{s}$ or $t(s)$ to mean ``the type of
s'', and denote type-transformed functions and variables using a
prime, e.g. \emph{s'}.

I will use the notation \emph{\{b\}} to indicate ``\emph{b} may have
to be transformed'' during the inference process.

\subsection{General Assumption on the Program}

The general assumption is that the program is built entirely out of 
\begin{itemize}
\item a set of functions on atomic types $f_{j}\,::\, a_{i}\rightarrow a_{k}$
\item the \emph{fold} operation \cite{hutton1999tutorial}
\item the \emph{cons} (\emph{:}) operation
\item tuple construction \emph{(,)}
\end{itemize}
However, we immediately note that \emph{foldl} can be defined in terms
of \emph{fold} and the identity function \emph{id}:

\[
foldl\, f\, v\, xs=foldr\,(\lambda x\, g\rightarrow(\lambda a\rightarrow g\,(f\, a\, x)))\, id\, xs\, v
\]

and \emph{map} in terms of \emph{fold} and \emph{cons}:

\[
map\, f=foldr\,(\lambda v\, x\rightarrow(x:v))
\]

and the same goes for all basic list operations, so we take those
as given.

Essentially, our purpose is to split the program in computational
functions and functions which describe the communication. Based on
the type transformations, we aim to derive the transformation of those
higher-order functions. We start with a few exploratory examples using
\emph{map} and \emph{foldl}.

\subsection{Increasing dimensionality -- map}

We assume a very simple program, we use Haskell syntax \cite{hudak1992report} augmented with
the \texttt{$\left\langle N\right\rangle $} notation to indicate
sizes.
\begin{lyxcode}
$s\,::\,[a]\left\langle N\right\rangle $

$g\,::\,[a]\left\langle N\right\rangle \rightarrow[a]\left\langle N\right\rangle $

$r\,::\,[b]\left\langle N\right\rangle $

$f\,::\, a\rightarrow a$

$r\,=\, g\, s$

$g\,=\, map\, f$
\end{lyxcode}
For completeness:
\begin{lyxcode}
$map\,::\,(t_{1}\rightarrow t_{2})\rightarrow[t_{1}]\left\langle n\right\rangle \rightarrow[t_{2}]\left\langle n\right\rangle $

\end{lyxcode}
We transform the top-level type:

$t_{s}'$\emph{$=R\, k\, M\, S\,$}$t_{s}$ \emph{$=R\, k\, M\, S\,[a]\left\langle N\right\rangle $$=[[a]\left\langle k\right\rangle ]\left\langle N/k\right\rangle $}

As \emph{g} is applied to \emph{s}, this leads to a transformation
of \emph{g}:

$t_{g'}$= $t_{s'}\rightarrow t?$

We assume that we only explicitly transform each of the arguments
of \emph{g}.

Then we get:
\begin{lyxcode}

$g'\,::\, t_{s}'\rightarrow[b]\left\langle N\right\rangle $

$g'\,=\, map\, f'$

\end{lyxcode}
We can substitute $t_{1}$ by the actual type of \emph{s' }in the
\emph{map} inside \emph{g'}:
\begin{lyxcode}
$r'\,=\, g'\, s'\,=\, map_{g'}\, f'\, s'$

$map_{g'}\,::\,\left(V^{-1}\,(N/k)\, t_{s'}\right)\rightarrow\left\{ b\right\} \rightarrow t_{s'}\rightarrow\left\{ [b]\left\langle N\right\rangle \right\} $

\end{lyxcode}
Clearly, this type can't work for \emph{map} because the return type
\emph{$[b]\left\langle N\right\rangle $} has a different size from
$t_{s'}$. So we need to transform that type:

$t_{2'}=R\, k\, M\, S\,[b]\left\langle N\right\rangle =[b\left\langle k\right\rangle ]\left\langle N/k\right\rangle $ 

This means that the signature for the map in \emph{g'} becomes
\begin{lyxcode}
$map_{g'}\,::\,\left(V^{-1}(N/k)\, t_{s'}\rightarrow V^{-1}(N/k)\, t_{2'}\right)\rightarrow t_{s'}\rightarrow t_{2'}$
\end{lyxcode}
So that
\begin{lyxcode}
$g'\,::\, t_{s'}\rightarrow t_{2'}$
\end{lyxcode}
Consequently
\begin{lyxcode}
$t_{r'}=t(g'\, s')=t_{2'}$

$\rightarrow$
\end{lyxcode}
Rewriting the above in a more systematic way:
\begin{lyxcode}
g'~s'~=~map~f'~s'

map$_{g'}$$\,::\,(t_{1}$~$\rightarrow$~$t_{2}$$)$~$\rightarrow$~$[t_{1}]$~$\rightarrow$~$[t_{2}]$

map$_{g'}$~::~(\emph{\{a\}}~->~\emph{\{{[}b{]}<k>\}})~->~{[}{[}a{]}<k>{]}<N/k>~->~\emph{\{{[}b<N>\}}

map$_{g'}$~::~(\emph{\{}$V^{-1}\, N/k$~\emph{{[}{[}a{]}<k>{]}<N/k>\}}~->~\emph{\{{[}b{]}<k>\}})~->~{[}{[}a{]}<k>{]}<N/k>~->~\emph{\{{[}b<N>\}}

map$_{g'}$~::~({[}a{]}<k>~->~\emph{\{{[}b{]}<k>\}})~->~{[}{[}a{]}<k>{]}<N/k>~->~\emph{\{{[}b{]}<N>\}}

map$_{g'}$~::~({[}a{]}<k>~->~\emph{\{{[}b{]}<k>\}})~->~{[}a{]}<k{*}m>~->~$R\, k\, M\, S$~\emph{\{{[}b{]}<N>\}}

map$_{g'}$~::~({[}a{]}<k>~->~\emph{\{{[}b{]}<k>\}})~->~{[}a{]}<k{*}m>~->~{[}b<k>{]}<N/k>

map$_{g'}$~::~({[}a{]}<k>~->~\emph{\{}$V^{-1}\, N/k$~\emph{{[}b<k>{]}<N/k>\}})~->~{[}a{]}<k{*}m>~->~{[}b{]}<k{*}m>

map$_{g'}$~::~({[}a{]}<k>~->~b<k>)~->~{[}a{]}<k{*}m>~->~{[}b{]}<k{*}m>

$\Rightarrow$~f'~::~{[}a{]}<k>~->~b<k>

$\Rightarrow$~r'~::~{[}b<k>{]}<N/k>
\end{lyxcode}
In other words, we can infer the return type from the single type
transformation. What we have so far is
\begin{lyxcode}
$s'\,::\,[a']\left\langle N'\right\rangle $

$g'\,::\,[a']\left\langle N'\right\rangle \rightarrow[b']\left\langle N'\right\rangle $

$r'\,::\,[b']\left\langle N'\right\rangle $

$f'\,::\, a'\rightarrow b'$

$r'\,=\, g'\, s'$

$g'\,=\, map\, f'$
\end{lyxcode}
where
\begin{lyxcode}
$\mbox{type}\, a'\,=\,[a]\left\langle k\right\rangle $

$\mbox{type}\, b'\,=\,[b]\left\langle k\right\rangle $

$N'\,=\, N/k$
\end{lyxcode}
What we need now is the transformations between \emph{s} and \emph{s'}
and \emph{f} and \emph{f'}

To transform \emph{s}:
\begin{lyxcode}
$s'\,=\,\mbox{reshapeTo}\, k\, s$
\end{lyxcode}
where
\begin{lyxcode}
$\mbox{reshapeTo}\,::\,\mbox{Int}\, k\Rightarrow k\rightarrow[a]\left\langle n\right\rangle \rightarrow[[a]\left\langle k\right\rangle ]\left\langle n/k\right\rangle $
\end{lyxcode}
We define the inverse for further use:
\begin{lyxcode}
$\mbox{reshapeFrom}\,::\,\mbox{Int}\, k\Rightarrow k\rightarrow[a\left\langle k\right\rangle ]\left\langle n\right\rangle \rightarrow[a]\left\langle n.k\right\rangle $
\end{lyxcode}
So the \emph{R k M S t(s)} type transformation maps directly to \emph{reshapeTo
k s}

The transformation from \emph{f} to \emph{f'} is even more straightforward,
because the transformation of the original type raises the dimensionality
\begin{lyxcode}
$f'\,=\, map\, f$
\end{lyxcode}
In general, the original map is replaced by maps over both dimensions.
\begin{lyxcode}
$g'\,=\, map\,map\, f$
\end{lyxcode}
and in full, the transformed program becomes

\begin{lyxcode}
$r\,=\, (reshapeFrom \,k) \; .\; (map\,map\, f) \;.\; (reshapeTo\,k)\,  s$
\end{lyxcode}

\subsection{Reducing the dimensionality -- map}

Assume we have
\begin{lyxcode}
\textrm{$s\,::\,[[a]\left\langle k\right\rangle ]\left\langle m\right\rangle $}

\textrm{$g\,::\,[[a]\left\langle k\right\rangle ]\left\langle m\right\rangle \rightarrow[[b]\left\langle k\right\rangle ]\left\langle m\right\rangle $}

$g\,=\, map\, f$

$f\,::\,[a]\left\langle k\right\rangle \rightarrow[b]\left\langle k\right\rangle $

$r\,=\, g\, s$
\end{lyxcode}
And we apply the transformation\emph{ M }$S^{-1}$ $R^{-1}$\emph{k}
to \emph{s}:

$R^{-1}\, k\, t(s)=[[a]\left\langle 1\right\rangle ]\left\langle k.m\right\rangle $

$M\, S^{-1}\,[[a]\left\langle 1\right\rangle ]\left\langle k.m\right\rangle =[a]\left\langle k.m\right\rangle $

\emph{$t(s')=[a]\left\langle k.m\right\rangle $}

So we obtain
\begin{lyxcode}
$s'\,::\,=[a]\left\langle k.m\right\rangle $
\end{lyxcode}
As \emph{g'} is applied to \emph{s'}, we obtain 
\begin{lyxcode}
$g'\,::\,[a]\left\langle k.m\right\rangle \rightarrow\left\{ [[b]\left\langle k\right\rangle ]\left\langle m\right\rangle \right\} $
\end{lyxcode}
Now we use inference on \emph{map}:
\begin{lyxcode}
g'~s'~=~map~f'~s'

map$_{g'}$~::~(t\_1~->~t\_2)~->~{[}t\_1{]}~->~{[}t\_2{]}

map$_{g'}$~::~(a~->~\emph{\{{[}b{]}<k>\}})~->~{[}a{]}<k{*}m>~->~\emph{\{{[}{[}b{]}<k>{]}<m>\}}

map$_{g'}$~::~(a~->~\emph{\{{[}b{]}<k>\}})~->~{[}a{]}<k{*}m>~->~\emph{M~}$S^{-1}$~$R^{-1}$\emph{~k~\{{[}{[}b{]}<k>{]}<m>\}}

map$_{g'}$~::~(a~->~\emph{\{{[}b{]}<k>\}})~->~{[}a{]}<k{*}m>~->~{[}b{]}<k{*}m>

map$_{g'}$~::~(a~->~\emph{\{}$V^{-1}\, k.m$~\emph{{[}b{]}<k{*}m>\}})~->~{[}a{]}<k{*}m>~->~{[}b{]}<k{*}m>

map$_{g'}$~::~(a~->~b)~->~{[}a{]}<k{*}m>~->~{[}b{]}<k{*}m>

$\Rightarrow$~f'~::~a~->~b

$\Rightarrow$~r'~::~{[}b{]}<k{*}m>
\end{lyxcode}
to express \emph{f'} as a function of \emph{f} , we need a \emph{toVector
k x} function
\begin{lyxcode}
$\mbox{toVector}\,::\,\mbox{Int}\, k\Rightarrow k\rightarrow a\rightarrow[a]\left\langle k\right\rangle $
\end{lyxcode}
The most intuitive implementation seems to be
\begin{lyxcode}
toVector~::~k~x~=~replicate~k~x
\end{lyxcode}
Similarly, we need \emph{fromVector k x} (although we don't really
need \emph{k})
\begin{lyxcode}
$\mbox{fromVector}\,::\,\mbox{Int}\, k\Rightarrow k\rightarrow[a]\left\langle k\right\rangle \rightarrow a$
\end{lyxcode}
The most intuitive implementation seems to be
\begin{lyxcode}
fromVector~k~(x:\_)~=~x
\end{lyxcode}
With these, we simply say
\begin{lyxcode}
$f'\, x\,=\,\mbox{fromVector}\, k\,\left(f\,\left(\mbox{toVector}\, k\, x\right)\right)$
\end{lyxcode}

\subsubsection{Correctness condition}

In general, the above transformation does not necessarily preserve
the computation. However, we can see that a sufficient condition to
preserves the computation is that\emph{ map f' = f}:

\begin{lemma}\label{ConditionOnReducedDim}

Mapping \emph{f'} over \emph{s'} preserves the computation of mapping
\emph{f} over \emph{s} iff 

\emph{f = map h}

\end{lemma}

\begin{proof}

~
\begin{enumerate}
\item Observe that \emph{s' = reshapeFrom k s} and we must show that

r' = \emph{g' s' = reshapeFrom k r = reshapeFrom k g s}

\item We show that \emph{map f' s' = map h s'}

\begin{enumerate}
\item Mapping \emph{f'} to \emph{s'}:

\emph{r' = g' s' = map f' s'}

\emph{= {[}f' x1,f' x2,...,f' xk,f' y1,f' y2,...,f' yk,...,f' z1,f'
z2,...,f' zk{]}}

\item \emph{f'} is identical to \emph{h}:

\emph{f' x}

\emph{= fromVector k (f (toVector k x))}

\emph{=head (f {[}x{]})}

\emph{= head (map h {[}x{]})}

\emph{= head {[}h x{]}}

\emph{= h x}

$\Rightarrow$\emph{f' = h}

\end{enumerate}

$\Rightarrow$\emph{ r' = map f' s' = {[}h x1,h x2,...,h xk,h y1,h
y2,...,h yk,...,h z1,h z2,...,h zk{]}}

\item Mapping \emph{f} to \emph{s}:

\emph{r = g s = map f s }

\emph{= map f {[}{[}x1,x2,...,xk{]},{[}y1,y2,...,yk{]},...,{[}z1,z2,...,zk{]}{]}}

\emph{= {[}f {[}x1,x2,...,xk{]},f {[}y1,y2,...,yk{]},...,f {[}z1,z2,...,zk{]}{]},...{]}}

\emph{= {[}map h {[}x1,x2,...,xk{]},map h {[}y1,y2,...,yk{]},...,map
h {[}z1,z2,...,zk{]}{]}}

\emph{= {[}{[}h x1,h x2,...,h xk{]}, {[}h y1,h y2,...,h yk{]},...,{[}h
z1,h z2,...,h zk{]}{]}}

\item Finally, transforming \emph{r} to \emph{r'}:

\emph{reshapeFrom k r }

\emph{= reshapeFrom k {[}{[}h x1,h x2,...,h xk{]}, {[}h y1,h y2,...,h
yk{]},...,{[}h z1,h z2,...,h zk{]}{]} }

\emph{= {[}h x1,h x2,...,h xk,h y1,h y2,...,h yk,...,h z1,h z2,...,h
zk{]} }

\emph{= r'}

\end{enumerate}
\end{proof}

\subsection{Preserving the dimensionality -- map}

With the same example as above, we apply the transformation\emph{
}$R\, n\, R^{-1}\, k$ to \emph{s}:

$R^{-1}\, k\, t(s)=[[a]\left\langle 1\right\rangle ]\left\langle k.m\right\rangle $

\emph{$R\, n\,[[a]\left\langle 1\right\rangle ]\left\langle k.m\right\rangle =[a\left\langle n\right\rangle ]\left\langle k.m/n\right\rangle $}

$t(s')=[a\left\langle n\right\rangle ]\left\langle k.m/n\right\rangle $

So we obtain

$s'\,::\,[a\left\langle n\right\rangle ]\left\langle k.m/n\right\rangle $

As \emph{g'} is applied to \emph{s'}, we obtain 

$g'\,::\,[a\left\langle n\right\rangle ]\left\langle k.m/n\right\rangle \rightarrow\left\{ [b\left\langle k\right\rangle ]\left\langle m\right\rangle \right\} $

Again we use inference on \emph{map}:
\begin{lyxcode}
g'~s'~=~map~f'~s'

map$_{g'}$~::~(t\_1~->~t\_2)~->~{[}t\_1{]}~->~{[}t\_2{]}

map$_{g'}$~::~(a<n>~->~\emph{\{{[}b{]}<k>\}})~->~{[}a<n>{]}<k{*}m/n>~->~\emph{\{{[}{[}b{]}<k>{]}<m>\}}

map$_{g'}$~::~(a<n>~->~\emph{\{{[}b{]}<k>\}})~->~{[}a{]}<k{*}m>~->~$R\, n\, R^{-1}\, k$\emph{~~\{{[}{[}b{]}<k>{]}<m>\}}

map$_{g'}$~::~(a<n>~->~\emph{\{{[}b{]}<k>\}})~->~{[}a{]}<k{*}m>~->~{[}b<n>{]}<k{*}m/n>

map$_{g'}$~::~(a~->~\emph{\{}$V^{-1}\, k.m/n$~\emph{{[}b<n>{]}<k{*}m/n>\}})~->~{[}a{]}<k{*}m>~->~{[}b{]}<k{*}m>

map$_{g'}$~::~(a<n>~->~b<n>)~->~{[}a{]}<k{*}m>~->~{[}b{]}<k{*}m>

$\Rightarrow$~f'~::~a<n>~->~b<n>

$\Rightarrow$~r'~::~{[}b<n>{]}<k{*}m/n>
\end{lyxcode}
As map is independent of the size of the vector, we have 
\begin{lyxcode}
$f'\,=\, f$
\end{lyxcode}
Consequently, the computation will always be preserved.

\subsection{Increasing dimensionality -- fold}

We can easily show that if the operation on f is a fold, then increasing
the dimensionality results in applying the fold to every dimension. 

\begin{lemma}\label{FoldOnNestedList}

Repeated application of \emph{fold} to a nested list is equivalent
to applying \emph{fold} to the flattened list

\emph{fold (fold f) acc {[}}\textcolor{blue}{\emph{{[}x1,x2,...xk{]}}}\emph{,}\textcolor{red}{\emph{{[}y1,y2,...yk{]}}}\emph{,...{]}
= fold f acc {[}x1,x2,...,xk,y1,y2,...,yk,...,z1,z2,...,zk{]}}

\end{lemma}

\begin{proof}

~

\emph{fold (fold f) acc {[}}\textcolor{blue}{\emph{{[}x1,x2,...xk{]}}}\emph{,}\textcolor{red}{\emph{{[}y1,y2,...y{]}}}\emph{,...,{[}z1,z2,...,zk{]}{]} }

\emph{= (fold f ... }\textcolor{red}{\emph{(fold f}}\emph{ }\textcolor{blue}{\emph{(fold
f acc {[}x1,x2,...,xk{]})}}\emph{ }\textcolor{red}{\emph{{[}y1,y2,...,yk{]})}}\emph{
... {[}z1,z2,...,zk{]})}

\emph{= (fold f ...}\textcolor{red}{\emph{(fold f}}\emph{ }\textcolor{blue}{\emph{(f
... (f (f acc x1) x2) ... xk) }}\textcolor{red}{\emph{{[}y1,y2,...,yk{]})}}\emph{
... {[}z1,z2,...,zk{]})}

\emph{= (f ... (f (f ( ... }\textcolor{red}{\emph{(f ... (f (f}}\emph{
}\textcolor{blue}{\emph{(f ... (f (f acc x1) x2) ... xk)}}\emph{ }\textcolor{red}{\emph{y1)
y2) ... yk)}}\emph{ ... ) z1) z2) ... zk)}

\emph{= fold f acc {[}x1,x2,...,xk,y1,y2,...,yk,...,z1,z2,...,zk{]}}

\end{proof}

Furthermore, as we consider a streaming operations, we only consider
the left fold (\emph{foldl}). 

We assume the same program as for \emph{map} above:
\begin{lyxcode}
$s\,::\,[a]\left\langle n\right\rangle $

$g\,::\,[a]\left\langle n\right\rangle \rightarrow b$

$r\,::\, b$

$f\,::\, b\rightarrow a\rightarrow b$

$acc\,::\, b$

$r\,=\, g\, s$

$g\,=\,\mbox{fold}\, f\, acc$
\end{lyxcode}
For completeness:
\begin{lyxcode}
$\mbox{fold}\,::\,(t_{2}\rightarrow t_{1}\rightarrow t_{2})\rightarrow t_{2}\rightarrow[t_{1}]\left\langle m\right\rangle \rightarrow t_{2}$

\end{lyxcode}
We transform the top-level type:

$t_{s}'$\emph{= R k M S }$t_{s}$ \emph{= R k M S {[}a{]}<n> = {[}{[}a{]}<k>{]}<n/k>}

So we obtain
\begin{lyxcode}
$s'\,::\,[a\left\langle k\right\rangle ]\left\langle n/k\right\rangle $
\end{lyxcode}
As \emph{g'} is applied to \emph{s'}, we obtain 
\begin{lyxcode}
$g'\,::\,[a\left\langle n\right\rangle ]\left\langle k.m/n\right\rangle \rightarrow\left\{ b\right\} $
\end{lyxcode}
Using inference on \emph{fold}:
\begin{lyxcode}
g'~s'~=~fold~f'~acc~s'

fold$_{g'}$~::~(t\_2~->~t\_1~->~t\_2)~->~t\_2~->~{[}t\_1{]}<m>~->~t\_2

fold$_{g'}$~::~(\emph{\{b\}}~->~\emph{\{a\}}->\emph{\{b\}})~->~\emph{\{b\}}~->~{[}a{]}<k><n/k>~->~\emph{\{b\}}

fold$_{g'}$~::~(\emph{\{b\}~->~\{}$V^{-1}\, n/k$~\emph{{[}a{]}<k><n/k>\}}~->~\emph{\{b\}})~->~{[}a{]}<k><n/k>~->~\emph{\{b\}}

fold$_{g'}$~::~(\emph{\{b\}~->}~{[}a{]}<k>~->~\emph{\{b\}})~->~{[}a{]}<k><n/k>~->~\emph{\{b\}}
\end{lyxcode}
At this point, the types are valid, so no further transformation is
required
\begin{lyxcode}
fold$_{g'}$~::~(b~->~{[}a{]}<k>~->~b)~->~{[}a{]}<k><n/k>~->~b

$\Rightarrow$~f'~::~b~->~{[}a{]}<k>~->~b

$\Rightarrow$~r'~::~b

$\Rightarrow$~acc'~::~b
\end{lyxcode}
To transform \emph{f} into \emph{f'}:
\begin{lyxcode}
$f'\,=\, fold\, f$
\end{lyxcode}

\subsection{Decreasing dimensionality -- fold}

We assume the same program as for map above:
\begin{lyxcode}
$s\,::\,[a\left\langle k\right\rangle ]\left\langle m\right\rangle $

$g\,::\,[a\left\langle k\right\rangle ]\left\langle m\right\rangle \rightarrow b$

$r\,::\, b$

$f\,::\, b\rightarrow a\left\langle k\right\rangle \rightarrow b$

$acc\,::\, b$

$r\,=\, g\, s$

$g\,=\,\mbox{fold}\, f\, acc$
\end{lyxcode}
For completeness:
\begin{lyxcode}
$\mbox{fold}\,::\,(t_{2}\rightarrow t_{1}\rightarrow t_{2})\rightarrow t_{2}\rightarrow[t_{1}]\left\langle m\right\rangle \rightarrow t_{2}$
\end{lyxcode}
We transform the top-level type:

$t_{s}'$\emph{= $M\, S^{-1}\, R^{-1}\, k$ }$t_{s}$ \emph{= $M\, S^{-1}\, R^{-1}\, k$
{[}a{]}<k><m> = {[}a{]}<k.m>}

So we obtain
\begin{lyxcode}
$s'\,::\,[a]\left\langle k.m\right\rangle $

$s'\,=\,\mbox{flatten}\, s$
\end{lyxcode}
As \emph{g'} is applied to \emph{s'}, we obtain 
\begin{lyxcode}
$g'\,::\,[a]\left\langle k.m\right\rangle \rightarrow\left\{ b\right\} $
\end{lyxcode}
Using inference on \emph{fold}:
\begin{lyxcode}
g'~s'~=~fold~f'~acc~s'

fold$_{g'}$~::~(t\_2~->~t\_1~->~t\_2)~->~t\_2~->~{[}t\_1{]}<m>~->~t\_2

fold$_{g'}$~::~(\emph{\{b\}}~->~\emph{\{{[}a{]}<k>\}}->\emph{\{b\}})~->~\emph{\{b\}}~->~{[}a{]}<k{*}m>~->~\emph{\{b\}}

fold$_{g'}$~::~(\emph{\{b\}~->~\{}$V^{-1}\, k.m$~\emph{{[}a{]}<k.m>\}}~->~\emph{\{b\}})~->~{[}a{]}<k{*}m>~->~\emph{\{b\}}

fold$_{g'}$~::~(\emph{\{b\}~->}~a~->~\emph{\{b\}})~->~{[}a{]}<k{*}m>~->~\emph{\{b\}}
\end{lyxcode}
At this point, the types are valid, so no further transformation is
required
\begin{lyxcode}
fold$_{g'}$~::~(b~->~a~->~b)~->~{[}a{]}<k{*}m>~->~b

$\Rightarrow$~f'~::~b~->~a~->~b

$\Rightarrow$~r'~::~b

$\Rightarrow$~acc'~::~b
\end{lyxcode}
To transform \emph{f} into \emph{f'}:
\begin{lyxcode}
$f'\, acc\, x\,=\, f\, acc\,\left(\mbox{toVector}\, k\, x\right)$
\end{lyxcode}

\subsubsection{Correctness condition}

In the case of fold, ``preserves the computation'' means ``produces
an identical result'', as from the perspective of the type transformation,
the type \emph{b} is opaque. In general, folding \emph{f'} over \emph{s'}
is not equal to folding \emph{f} over \emph{s}. However, a sufficient
condition for equality is this:

\begin{lemma}

Folding \emph{f'} over \emph{s'} is equal to folding \emph{f} over
\emph{s} iff 

\emph{f = fold h}

\end{lemma}

\begin{proof}

~
\begin{enumerate}
\item \emph{foldl f' acc s' = foldl h acc s'}

\emph{foldl f' acc {[}x1,x2,...,xk,y1,y2,...,yk,...,z1,z2,...,zk{]}}

{[}def. of f'{]}

\emph{= foldl (\textbackslash{}acc x -> f acc (toVector k x)) acc
{[}x1,x2,...,xk,y1,y2,...,yk,...,z1,z2,...,zk{]}}

{[}def of toVector{]}

\emph{= foldl (\textbackslash{}acc x -> f acc {[}x{]}) acc {[}x1,x2,...,xk,y1,y2,...,yk,...,z1,z2,...,zk{]}}

{[}def of foldl{]} foldl h acc {[}x{]} = h acc x

\emph{= foldl (\textbackslash{}acc x -> h acc x) acc {[}x1,x2,...,xk,y1,y2,...,yk,...,z1,z2,...,zk{]}}

{[}$\eta$ conversion{]}

\emph{= foldl h acc {[}x1,x2,...,xk,y1,y2,...,yk,...,z1,z2,...,zk{]}}

{[}def. of foldl{]}

\emph{= (h (... (h (h (... (h ... (h (h (fh (h ... (h (h acc x1) x2)
... xk) y1) y2) ... yk) ...) z1) z2) ...) zk)}

\item \emph{foldl h acc s' = foldl f acc s}

\emph{foldl f acc s}

{[}def. of f{]}

\emph{= foldl (foldl h) acc s}

{[}Lemma \ref{FoldOnNestedList} + def. of s'{]}

\emph{= foldl h acc s'}

\end{enumerate}
\end{proof}

\subsection{About \emph{zip} and \emph{unzip} }

We use \emph{zip} and \emph{unzip} to change nested lists of tuples
into tuples of nested lists. 
\begin{lyxcode}
zip~::~{[}a{]}<n>~->~{[}b{]}<n>~->~{[}(a,b){]}<n>

unzip~::~{[}(a,b){]}<n>~->~({[}a{]}<n>,{[}b{]}<n>)
\end{lyxcode}
The same type transformation must be applied to both arguments, e.g.
for \emph{R k {[}a{]}<n>}. In order to preserve the computation, it
is quite clear that
\begin{lyxcode}
zip'~::~{[}{[}a{]}<k>{]}<n/k>~->~{[}{[}b{]}<k>{]}<n/k>~->~{[}{[}(a,b){]}<k>{]}<n/k>
\end{lyxcode}
can be implemented in terms of \emph{zip} as
\begin{lyxcode}
zip'~xs'~ys'~=~map~(\textbackslash{}(x,y)~->~zip~x~y)~(zip~xs'~ys')
\end{lyxcode}
and similar for \emph{unzip}.

To simplify the discussion, we introduce a variant of \emph{zip},
\emph{zipt}, which takes a single tuple as argument, and a corresponding
\emph{unzipt}.
\begin{lyxcode}
zipt~::~({[}a{]}<n>,{[}b{]}<n>)~->~{[}(a,b){]}<n>

zipt~(xs,ys)~=~zip~xs~ys
\end{lyxcode}
and
\begin{lyxcode}
unzipt~::~{[}(a,b){]}<n>~->~({[}a{]}<n>,{[}b{]}<n>)

unzipt~ltups~=~(map~fst~ltups,~map~snd~ltups)
\end{lyxcode}
then \emph{zipt'} becomes 
\begin{lyxcode}
zipt'~::~({[}a{]}<k><n/k>,{[}b{]}<k><n/k>)~->~{[}(a,b){]}<k><n/k>
zipt'~tup~=~map~zipt~(zipt~tup)
\end{lyxcode}
and similar for \emph{unzipt}.

\section{Conclusion}

The approach described allows to transform programs consisting of
combinations of \emph{map}, \emph{foldl} and \emph{zip} based on transformation of the types
of the vectors on which the map or fold acts. 

We have shown that the the set \emph{V(a,n)} of vectors of type $a$ and size $n$ is closed under the proposed operations for transforming the vector types, \emph{S},\emph{M} and \emph{R}, with the corollary that every combination of the transformations is reversible.

We have shown that, for programs consisting of opaque functions and the operations \emph{map}, \emph{foldl} and \emph{zip}, the program transformations can be automatically derived from the type transformations.

This mechanism allows to generate correct-by-construction variants of the programs. The purpose of this works is to allow automatic selection of the variant most suitable for a given platform through optimisation against a platform cost model.

This work is supported by the EPSRC through the TyTra project (EP/L00058X/1).

\bibliographystyle{apalike}
\bibliography{infering_programs_from_type_transformations}

\end{document}